\documentclass[12pt]{article}
\usepackage{amsmath,amssymb}
\usepackage{amsthm}

\def\LC{\mathcal{L}}
\def\TC{\mathcal{T}}

\def\FC{\mathcal{F}}
\def\HC{\mathcal{H}}

\newcommand{\E}{\mathbf{E}}
\newcommand{\N}{\mathbf{N}}
\def\P{\mathbf{P}}
\def\Q{\mathbf{Q}}
\newcommand{\R}{\mathbf{R}}

\newcommand{\e}{\mathbf{e}}

\def\tr{\rm{tr}}

\def\arccot{\rm{arccot}}

\newcommand{\al}{\alpha}
\newcommand{\be}{\beta}
\newcommand{\pa}{\partial}
\newcommand{\ep}{\epsilon}

\newcommand{\ga}{\gamma}

\newtheorem{prop}{Proposition}[section]
\newtheorem{theorem}{Theorem}[section]

\newtheorem{corollary}{Corollary}
\newtheorem{remark}{Remark}

\newcommand{\la}{\lambda}
\newcommand{\si}{\sigma}

\newcommand{\om}{\omega}

\newcommand{\De}{\Delta}

\newcommand{\Om}{\Omega}

\begin{document}
\title{On the Mathematical Theory of Quantum Stochastic Filtering Equations for Mixed States
\thanks{The paper was supported by RSF research grant (project no. 24-11-00123)}
}

\author{Vassili N. Kolokoltsov\thanks{Faculty of Computation Mathematics and Cybernetics,
 Moscow State University, 119991 Moscow, Russia, Higher School of Economics, Russia,
 Professor Emeritus of the University of Warwick,
  Email: v.n.kolokoltsov@gmail.com}}
\maketitle

\begin{abstract}
Quantum filtering equations for mixed states were developed in 80th of the last 
century. Since then the problem of building a rigorous mathematical theory for these equations in the basic infinite-dimensional settings has been a challenging open mathematical problem. In a previous paper, the author developed the theory of these equations in the case of bounded coupling operators, including a new version that arises as the law of large numbers for interacting particles under continuous observation and thus leading to the theory of quantum mean field games. In this paper, the main body of these results is extended to the basic cases of unbounded coupling operators.   
\end{abstract}



\section{Introduction}

\subsection{Aims and scope}

The general theory of quantum non-demolition observation, filtering, and resulting
feedback control was built essentially by V.P. Belavkin in papers \cite{Bel87}, \cite{Bel88}, \cite{Bel92}, see review in \cite{BoutHanJamQuantFilt} and alternative simplified derivations in \cite{BelKol}, \cite{Pellegrini}, \cite{BarchBel}, \cite{Holevo91}, \cite{KolQuantFrac}, \cite{BarndLoub}, \cite{Loubenets}.
For the technical side of organizing feedback quantum control in real time
 see e.g. \cite{Armen02Adaptive}, \cite{Bushev06Adaptive} and \cite{WiMilburnBook}.
Equations of quantum filtering can also be looked at as stochastic master (or Lindblad) equations, see a detailed discussion in \cite{BarchBook}. We can refer also to \cite{Partha}
for various insightful links of these equations with other models.

The mathematics of the evolution of pure states given by the Belavkin equations (and representing some kind of stochastic nonlinear Schr\"odinger equation) is fairly well understood by now, though the issues related to the strong solutions of nonlinear equation were still open and we fix them bypassing here. The present paper is mainly devoted to the study of a more subtle case of the operator-valued evolution of mixed states. In the previous paper \cite{K} the author developed the theory of these equations under the assumption of boundedness of the coupling operator $L$ with a measurement device. Here these results are extended to the basic cases of unbounded coupling operators.

 Recently the author built the theory of the law of large numbers limit for continuously observed interacting quantum particle systems leading to quantum mean-field games, see \cite{KolQuantLLN}, \cite{KolQuantMFGCount} and \cite{KolQuantMFG}. These limits are described by certain nontrivial extensions of quantum stochastic master equations that can be looked at as infinite-dimensional operator-valued McKean-Vlasov diffusions. In \cite{K} the full theory of these new equations was built for bounded coupling operators. The result of the present paper allows one to extend them for unbounded $L$, but we do not touch this issue here, planning to consider it elsewhere.  On the other hand, the theory developed here, can be used to give a rigorous derivation of the quantum filtering equations from the limits of appropriately scaled sequences of instantaneous measurements, which were performed until now only for finite-dimensional situations, see e.g.  \cite{KolQuantFrac} and \cite{Pellegrini} and references therein. This issue will be also addressed elsewhere. 

\subsection{Main objects of analysis}

In this paper, the letters $H$ and $L=(L^1, \cdots, L^n)$ denote linear operators in a separable Hilbert space $\HC$, where $H$ is self-adjoint and is referred to as
the Hamiltonian. The vector-valued $L$ is closed and densely defined with densely defined adjoint $L^*$. It describes the coupling of a quantum system with a measurement device. 
We shall assume also that the operator $\sum_k (L^k)^*L^k$ is densely defined and closed 
(and consequently symmetric). 

We use the notations
\[
L_S=(L+L^*)/2=(L^{S1}, \cdots, L^{Sn}), \quad L_A=(L-L^*)/2i=(L^{A1}, \cdots, L^{An}).
\]
We write $\, Re \, z=z_R$ and $Im \, z=z_I$ for the real and imaginary parts of a complex number or a vector, 
write $Dom(A)$ for the domain of an operator $A$, use the brackets $[A,B]$ 
and $\{A,B\}$ to denote the commutator and anti-commutator of operators $A,B$. 
By $\|A\|$ we denote
the standard operator norm of a bounded operator $A$ in $\HC$. The norm of an operator $A$ from a Banach space $B_1$ to a Banach space $B_2$ will be denoted $\|A\|_{B_1\to B_2}$
and simply $\|A\|_B$ for operators $B\to B$. The letter $\E$ is used to denote the expectation, and we write $\E_{\P}$ to stress that the expectation is taken with respect to the measure $\P$.

\begin{remark}
Most of our results have an extension to the case of 
time dependent families $L(t)$, but we do not give details here.
\end{remark}

The quantum filtering equation  describing 
the stochastic evolution of pure states
(vectors in a Hilbert space $\HC$) under continuous measurement
of a diffusive type can be written in two versions:

(i) as the {\it linear Belavkin quantum filtering SDE (stochastic differential equation)} for a non-normalized state:
\begin{equation}
\label{eqqufiBlin}
d\chi(t) =-[iH\chi(t) +\frac12 \sum_{j=1}^n (L^j)^*L^j \chi(t) ]\,dt
+\sum_jL^j\chi(t) dY_j(t),
\end{equation}
where $\chi(t)\in \HC$ and $Y(t)=(Y_1, \cdots, Y_n)(t)$ is an $n$-dimensional Brownian motion (BM) referred to as the {\it output process};

(ii) as the {\it nonlinear Belavkin quantum filtering SDE}
 for the normalized state $\phi(t)$:
\[
d\phi(t)=\sum_j(L^j-(\phi(t), L^{Sj} \phi(t)))\phi(t) \, dB_j(t)
\]
\begin{equation}
\label{eqqufiBnonlin}
-[i(H-\sum_j(\phi(t), L^{Sj} \phi(t)) L^{Aj})
+\frac12\sum_{j=1}^n (L^j-(\phi(t), L^{Sj} \phi(t)))^*(L^j-(\phi(t), L^{Sj} \phi(t)))]\phi(t) \, dt,
\end{equation}
where $B(t)=(B_1, \cdots, B_n)(t)$ is an $n$-dimensional Brownian motion 
(BM) referred to as the {\it innovation process}.

Of course the rigorous form of these equations is the integral one. Say, equation \eqref{eqqufiBlin} with initial vector $\chi_0$ writes down as  
\begin{equation}
\label{eqqufiBlinint}
\chi(t)=\chi_0 +\int_0^t (-iH-\frac12 \sum_{j=1}^n (L^j)^*L^j) \chi(s)ds 
+\sum_{j=1}^n \int_0^t L^j \chi(s) \, dY_j(s).
\end{equation}

\begin{remark} In the literature one finds also a version of these equations with an infinite number of $L^j$. We decided here to work with finite $n$ to make our exposition more transparent (avoiding just additional technical complications). 
\end{remark}
    
Equation \eqref{eqqufiBlin} is clearly well posed in the case of bounded $H$ and $L$. Moreover, one checks by direct application of Ito's formula that if $\chi$ satisfies \eqref{eqqufiBlin}, then 
 \begin{equation}
\label{chisquare}
d\|\chi(t)\|^2=2\sum_j(\chi(t),L^{Sj} \chi(t))dY_j(t),
\end{equation}
and the normalized vectors
$\phi(t)=\chi(t)/\|\chi(t)\|$ satisfy \eqref{eqqufiBnonlin} with 
\begin{equation}
\label{eqdefinnov}
dB_j(t)=dY_j(t)-2(\phi(t), L^{Sj} \phi(t)) \, dt.
\end{equation}
Thus, by Girsanov's theorem, if $Y(t)$ is a BM on a certain stochastic basis 
$(\Om, \FC, \FC_t,\P)$, then $B(t)$ is a BM on the basis $(\Om, \FC,\FC_t,\Q)$ s.t.
 \begin{equation}
\label{chisqdensity}
\E_{\Q} \xi =\E_{\P}(\|\chi(t)\|^2 \xi)
\end{equation}
for $\FC_t$-measurable $\xi$.

If one agrees to understand all products of operator expressions
as appropriate inner products
(sum over available indices), for instance, writing $L^*L$ instead of $\sum_j (L^j)^* L^j$,
and $L \chi \, dY(t)$ instead of $\sum_j L^j\chi \, dY_j(t)$,
 one can write all equations above in a simpler form,
say the main equations \eqref{eqqufiBlin} and \eqref{eqqufiBnonlin} will look like
\begin{equation}
\label{eqqufiBlins}
d\chi(t) =-[iH\chi(t) +\frac12 L^*L \chi(t) ]\,dt+L \chi(t) dY(t),
\end{equation}
and, respectively,
\[
d\phi(t)=-[i(H-(\phi(t), L_S \phi(t)) L_A)
+\frac12 (L-(\phi(t), L_S \phi(t)))^*(L-(\phi(t), L_S \phi(t)))]\phi(t) \, dt
\]
\begin{equation}
\label{eqqufiBnonlins}
+ (L-(\phi(t), L_S \phi(t)))\phi(t) \, dB(t).
\end{equation}
We will mostly use this short way of writing having in mind more detailed versions above.

The derivation of \eqref{eqqufiBnonlins} from the linear equation, given above,
suggests that it is meant to describe evolutions of vectors with a unit norm. 
It is often convenient to include it in a more general class of norm-preserving evolutions, the simplest version being
\[
d\phi(t)=-[i(H-\langle L_S \rangle_{\phi(t)} L_A)
+\frac12 (L-\langle L_S \rangle_{\phi(t)})^*(L-\langle L_S \rangle_{\phi(t)})]\phi(t) \, dt
\]
\begin{equation}
\label{eqqufiBnonlinsn}
+ (L-\langle L_S \rangle_{\phi(t)})\phi(t) \, dB(t),
\end{equation}
where we introduced the (rather standard) notation for the value
of an operator $A$ in a pure state $\phi$:
\[
\langle A \rangle_{\phi}=\frac{(\phi, A \phi)}{(\phi, \phi)}.
\]
Clearly for $\phi(t)$ of unit norm solutions to equations \eqref{eqqufiBnonlins}
and \eqref{eqqufiBnonlinsn} coincide, but \eqref{eqqufiBnonlinsn} is explicitly norm-preserving for arbitrary $\phi(t)$, which is not the case for equation \eqref{eqqufiBnonlins}.

Equation \eqref{eqqufiBlins} simplifies essentially in the most important case of self-adjoint $L$, as it takes the form  

\begin{equation}
\label{eqqufiBnonlinsnsa}
d\phi(t)=-[iH
+\frac12 (L-\langle L \rangle_{\phi(t)})^2]\phi(t) \, dt
+ (L-\langle L \rangle_{\phi(t)})\phi(t) \, dB(t).
\end{equation}

It is also insightful to write down equation for $\chi$ 
in terms of the innovation process $B$:
\begin{equation}
\label{eqqufiBlinsB}
d\chi(t) =-[iH\chi(t) +\frac12 L^*L \chi(t) ]\,dt
+L \chi(t) \, (dB(t)+\langle L+L^*\rangle_{\chi(t)} \, dt).
\end{equation}

For unbounded $H,L$, the meaning of all these equations is of course not obvious.


Recall that the density matrix or density operator $\ga$ corresponding to a unit vector
$\chi\in \HC$ is defined as the orthogonal projection operator on $\chi$.
This operator is usually expressed either as the tensor product 
$\ga=\chi\otimes \bar \chi$ (with the usual identification of $H\otimes H$ 
with the space of Hilbert-Schmidt operators in $\HC$) or in the most common 
for physics bra-ket Dirac's notation as $\ga=|\chi\rangle \langle \chi|$.

As one checks by direct application of Ito's formula, (i) if $\chi(t)\in \HC$
satisfies \eqref{eqqufiBlins}, then the corresponding operator $\ga=\chi\otimes \bar \chi$
satisfies the {\it linear stochastic quantum master equation} or {\it linear Belavkin's quantum filtering equation for mixed states}
\begin{equation}
\label{Lindstoch}
d\ga(t)=-i[H,\ga(t)] \, dt +\LC_L \ga(t) \, dt +(L\ga(t)+\ga(t) L^*) dY(t),
\end{equation}
with
\[
\LC_L\ga =L\ga L^*-\frac12 L^*L\ga -\frac12 \ga L^*L
=L\ga L^*-\frac12 \{L^*L,\ga\};
\]
and (ii) if $\phi(t)$ satisfies \eqref{eqqufiBnonlin}, then the corresponding matrices
$\rho=\phi\otimes \bar \phi$ satisfies the {\it nonlinear stochastic quantum master equation}
or {\it nonlinear Belavkin's quantum filtering equation}
\begin{equation}
\label{Lindstochnorm1}
d\rho(t)=-i[H,\rho(t)]\, dt+\LC_L \rho (t)\, dt
+[L\rho(t)+\rho(t) L^*-\rho(t)\, {\tr} \, (L\rho(t)+\rho(t) L^*) ] dB(t).
\end{equation}

The well posedness of these master equations (not necessarily for solutions of form $\chi \otimes \bar \chi$) and of their nonlinear extensions for interacting particle systems were proved in \cite{K} in case of bounded $L$. The objective of the present paper is to extend the main body of these results to (practically important) cases of unbounded $L$. Notice that the arguments of the present paper use the ideas from \cite{K}, but the results are essentially independent, because the assumption of boundedness of $L$ allows for certain specific estimates that are absent in the unbounded case. 

As an important property of the stochastic equations above, let us observe that taking expectation from both sides of \eqref{Lindstoch} shows (at least heuristically in infinite dimensional case) that $\hat \ga (t)=\E \ga (t)$ satisfies the {\it quantum master equation}
or {\it Lindblad equation}
\begin{equation}
\label{eqquantmas}
\frac{d}{dt} \hat \ga(t)=-i[H,\hat \ga(t)]  +\LC_L \hat \ga(t).
\end{equation}
Solutions to these equations generating quantum dynamic semigroups were attentively studied in the literature, see \cite{Mora13} and numerous references therein.  

\subsection{Content}

The paper is organized as follows. 

In Section \ref{secpureeq}  we first recall some  known facts
on the quantum filtering SDEs (stochastic differential equations) for pure states, often providing  simplified proofs. We discuss two approaches, a more abstract one and another based on some explicit solutions, the latter giving more insight in the behavior of the stochastic dynamics. Finally we prove our first result concerning the well-posedness of the nonlinear Belavkin's equation for pure states in the strong (probabilistic) sense, while previously only probabilistically weak solutions were developed.  

In Section \ref{seclineq} we build the theory of the linear quantum filtering SDE \eqref{Lindstoch}. We follow the approach from \cite{K} considering these equations in the Hilbert space of Hilbert-Schmidt operators. This approach leads necessarily to SDEs with singular coefficients, but allows one to avoid working with the inconvenient (from the point of view of stochastic calculus) Banach space of trace-class operators. By passing we note that the well known result on the well-posedness and conservativity of dynamic quantum semigroups given by equation \eqref{eqquantmas}, is a direct consequence of our present theory on the corresponding quantum stochastic equations.    

In Section \ref{secnormeq}, armed with the results of the previous sections, we address the most nontrivial object of this paper: nonlinear quantum filtering equations 
for mixed states  \eqref{Lindstochnorm1}, and obtain the full well-posedness result for these equations, first in the weak (probabilistically) and then in the strong sense.

\section{Dynamics of pure states}
\label{secpureeq}

\subsection{Linear equation: general results}

In this section we collect some mostly known facts about
the linear quantum filtering equations for pure states. We shall present the well posedness of the linear Belavkin equation with unbounded coefficients in two settings: (1) the most general one due to Mora-Rebolledo based on the introduction of the so called control operator and with ideas going back to Holevo and Fagnola-Chebotarev) and (2) a more specific one arising from explicit Green functions of the standard stochastic Schr\"odinger equation describing continuously observed momentum or position.    

Starting with the first approach we introduce the following 
Hypothesis due to Mora-Rebolledo, though in a slightly reduced (and more intuitive) form,
which appears in all applications.

{\it Hypothesis MR}. Suppose there exists a (strictly) positive self-adjoint linear operator $C : \HC \to \HC$ with discrete spectrum $0\le \la_1 \le \la_2 \le \cdots $, the corresponding orthonormal basis of eigenvectors $\{e_m\}$, subspaces $\HC_m$ generated by the first $m$ eigenvectors and orthogonal projections $P_m$ on $\HC_m$. On the domain $Dom(C)$ of $C$ one can define the corresponding $C$-norm 
$\|x\|_C^2=(\|x\|^2 +\|Cx\|^2)^{1/2}$. It is then assumed that $Dom(C) \subset Dom(H) \cap Dom(L) \cap Dom (L^*L)$  and therefore (according to the closed graph theorem) there exists a constant $K$ such that
\begin{equation}
\label{MR0}
 \|H x\|^2\le K\|x\|^2_C, \quad 
 \| L^*L x\|^2 \le K\|x\|_C^2,   
\end{equation}
for all $x\in Dom(C)$, and consequently
\begin{equation}
\label{MR0a}
\|L x\|^2 \le \sqrt K \|x\| \, \|x\|_C.
\end{equation}
Finally, the following generalized {\it dissipativity with respect to} $C$ is assumed:  either
\begin{equation}
\label{MR1}    
-2Re \, (Cx,i CP_m Hx)-Re \, (Cx,CP_m L^*L x)
+ \|CP_mLx\|^2\le \al (\|x\|_C^2+\be),
\end{equation}
or (this second version is noted bypassing as Remark 21 in \cite{MoraRebolin})
\begin{equation}
\label{MR2}    
-2Re \, (Cx,iCP_m H x)-Re \, (Cx,CP_m L^*P_mL x)
+ \|CP_mLx\|^2\le \al (\|x\|_C^2+\be),
\end{equation}
hold for any $m$, $x\in \HC_m$ and some constants $\al, \be>0$.

Let us stress again that we use our reduced notations with the summation tacitly assumed, 
so that, e.g., $L^*L$ means $\sum_k (L^k)^*L^k$ and $\|Lx\|^2=\sum_k \|L^kx\|^2$.

Conditions \eqref{MR1} and \eqref{MR2} are similar and often hold simultaneously. However, \eqref{MR1} is usually easier to check, because 
it is equivalent (under all other conditions) to the inequality
\begin{equation}
\label{MR1a}    
-2Re \, (Cx,iC Hx)-Re \, (Cx,C L^*L x)
+ \|CLx\|^2\le \al (\|x\|_C^2+\be),
\end{equation}
holding for all $x\in \cup_m \HC_m$ (not involving any projections). In fact,
\eqref{MR1a} follows from \eqref{MR1} passing to the limit as $m\to \infty$, and 
\eqref{MR1a} implies \eqref{MR1}, because 
\[
-2Re \, (Cx,iCP_m Hx)-Re \, (Cx,CP_m L^*L x)
\]
\[
=-2Re \, (Cx,iC H x)-Re \, (Cx,C L^*L x), \quad x\in \HC_m.
\]
On the other hand,  \eqref{MR2} is more convenient for working with nonlinear equations, 
as will be seen later on.    

Given BM $Y(t)$ on a filtered probability space $(\Om, \FC, \FC_t, \P)$,
a $C$-{\it strong solution} of equation \eqref{eqqufiBlinint} on a time-interval $[0,T]$ is defined as an adapted continuous $\HC$-valued process $\chi(t)$ s.t. (i) its norm is non-increasing, that is, $\E \|\chi(t)\|^2\le \|\chi_0\|^2$, (ii) $\chi(t)\in Dom(C)$ a.s. and
\[
\sup_{t\in [0,T]} \E \| C \chi(t)\|^2 <\infty,
\]
(iii) equation \eqref{eqqufiBlinint} holds.

\begin{remark}
To be more precise, one introduces the mapping $\pi_C : \HC \to \HC$ s.t. $\pi_C(x)$ equals $x$ or $0$ for $x\in Dom(C)$ and otherwise, respectively, and one writes $\pi_C(\chi(t))$ instead of $\chi(t)$ everywhere in the equation. We shall assume that this is done whenever necessary, but we shall not explicitly use this $\pi_C$ in our notation.        
\end{remark}

The following result is the combination of Theorems 4 and 17 from \cite{MoraRebolin}
(given here under some simplifying assumption, e.g. assuming only a finite number of $L_j$).

\begin{theorem} 
\label{thMorRebLin}
(Mora-Rebolledo) Under Hypothesis MR, for any $\chi_0\in Dom(C)$ there exists a unique $C$-strong solution $\chi(t)$ to equation \eqref{eqqufiBlinint}, and moreover, this solution satisfies the estimate
\begin{equation}
\label{eqthMor1}
\E \|C\chi(t)\|^2 \le e^{\al t} [\|C\chi_0\|^2+\al t(\|\chi_0\|^2+\be)],
\end{equation}
and is conservative in the sense that $\E \|\chi(t)\|^2=\|\chi_0\|^2$, 
so that $\|\chi(t)\|^2$ is a positive martingale. 
\end{theorem}

\begin{remark}
 Theorem  \ref{thMorRebLin} extends several previous contributions by other authors (see e.g.  \cite{BarchHol}, \cite{Holevo96} and references therein), some of them being inspired by the works on the conservativity of quantum dynamic semigroups from \cite{ChebFagn} and \cite{ChebQuez}.
For other classes  of stochastic Schr\"odinger equation we can refer to \cite{BarbRock16} and references therein.  
\end{remark}

Paper \cite{MoraRebolin} contains a detailed non-trivial and elegant proof. We shall sketch here an essentially simplified argument assuming an additional condition that usually holds in examples.

{\it Hypothesis A}: (a) There exists $l\in \N$ such that $L:\HC_m\to \HC_{m+l}$, 
$H:\HC_m\to \HC_{m+l}$ for all $m$ and (b)   $Dom(\sqrt C) \subset Dom(H)$, and therefore (according to the closed graph theorem) there exists a constant $K'$ such that
\begin{equation}
\label{MR0my}
 \|H x\|^2\le K'\|x\|^2_{\sqrt C} \le 2K'\|x\| \, \|x\|_C,   
\end{equation}

\begin{proof} 
 First of all, uniqueness of solution is more or less straightforward under Hypothesis MR:
 one writes down the equation for the difference $\De(t)$ of two solutions and observes that $\|\De(t)\|^2$ is a positive local martingale (here is an important detail to note: it is not seen apriori that it is a martingale), then one introduces the stopping time $\tau_n$ when $\|\De(t)\|$ reaches the level $n$ and concludes that $\De_{t\wedge \tau_n}=0$ for all $n$, and then finally that $\De(t)=0$.
 
 The main point is really the existence, which is proved via approximations. Namely, under \eqref{MR1}, one considers the solutions to the equation
\begin{equation}
\label{eqMRappr1}
\chi_m(t)=P_m\chi_0 +\int_0^t P_m  (-iH-\frac12  L^* L) \chi_m(s)ds 
+ \int_0^t P_m L \chi_m(s) \, dY(s),
\end{equation}
and, under \eqref{MR2}, one considers the solutions to the equation
\begin{equation}
\label{eqMRappr2}
\chi_m(t)=P_m\chi_0 +\int_0^t P_m (-iH-\frac12  L^*P_m L) \chi_m(s)ds 
+ \int_0^t P_m L \chi_m(s) \, dY(s),
\end{equation}
in $\HC_m$. Everything is well-posed in this finite-dimensional setting. Moreover, 
writing down the equations for $C\chi_m$ (say, under \eqref{eqMRappr2}) one obtains 
\[
d \, \E \, \|C\chi_m\|^2 =2 \, Re\, (CP_m (-iH-\frac12 L^*P_m L)\chi_m, CP_m\chi_m) \, dt
+(CP_mL\chi_m, CP_m L\chi_m) \,dt.
\]
By Hypothesis MR and Gronwall's Lemma, it follows  
that $\chi_m$ satisfies \eqref{eqthMor1} for all $m$: 
\begin{equation}
\label{eqthMor2}
\E \|C\chi_m(t)\|^2 \le e^{\al t} [\|C\chi_0\|^2+\al t(\|\chi_0\|^2+\be)].
\end{equation}

So the main point in the proof is to show that $\chi_m$ (or its subsequence) converges (in some sense), as $m\to \infty$, and that the limit is the solution required.

\begin{remark}
The proof of \cite{MoraRebolin} goes as follows. First it is shown that there exists a weakly converging subsequence $\chi_{m_k}$ (in the Hilbert space of square integrable $\HC$-valued   random variables), then by accurate consideration of duality it is shown that the limit satisfies equation  \eqref{eqqufiBlinint}, and then rather elaborate stopping arguments (which are improved versions of the arguments from \cite{Holevo96}) allows one to show that this solution is conservative, and finally one concludes that subsequence $\chi_{m_k}$ converges strongly. We shall prove here directly the strong convergence of $\chi_m$, even including the rates of convergence, though assuming additionally Hypothesis A.
\end{remark}

We shall work with approximations \eqref{eqMRappr2}. This equation, considered in space $\HC_m$, coincides with the basic equation     
\eqref{eqqufiBlinint}, if one chooses as operators $H$ and $L$ their finite-dimensional approximations $H_m=P_mHP_m$ and $L_m=P_mLP_m$, respectively. As was mentioned, finite-dimensional equations of this kind are well understood and it is known that they are conservative: $\E\|\chi_m(t)\|^2=\|\chi_0\|^2$. In fact, to see that this is the case, one writes down the equation for $\|\chi\|^2$ and observes that it is a martingale.

Subtracting equations for $m$ and $k$ we obtain (omitting argument $t$ for brevity)
\[
d(\chi_k-\chi_m)=-[i(H_k-H_m) +\frac12  (L_k^*L_k-L_m^*L_m)]\chi_k\, dt
+(L_k-L_m)\chi_k \, dY(t) 
\]
\[
-(i H_m +\frac12 L_m^*L_m)(\chi_k-\chi_m)\, dt
+L_m (\chi_k-\chi_m) \, dY(t). 
\]
And therefore
\[
d \, \E \, (\chi_k-\chi_m)^2
=-2 \, \E \, (\chi_k-\chi_m, i(H_k-H_m) \chi_k)_R \, dt 
\]
\[
+\E \, (\chi_k-\chi_m, (L_k^*L_k-L_m^*L_m)\chi_k)_R \, dt
+\E \, \|(L_k-L_m)\chi_k\|^2 \, dt 
\]
\begin{equation}
\label{eqthMor3}
+2\, \E \,  ((L_k-L_m) \chi_k, L_m (\chi_k-\chi_m))_R \, dt,
\end{equation}
because the terms not containing differences $L_k-L_m$ or $H_k-H_m$ nicely cancel out.

To estimate the right hand side we need the following observation. If $C$ is as assumed in Hypothesis MR and $A$ is an operator such that $A:\HC_m\to \HC_{m+l}$ for all $m$ and $\|Ax\|^2\le R \|x\|\, \|x\|_C$ for $x\in Dom(C)$ and a constant $R$, then 
\begin{equation}
\label{eqthMor4}
\|(A_k-A_m) x\|\le  \frac{\sqrt {2R}}{\sqrt {\la_{m-l+1}}}  \|x\|_C, 
\quad
\|(A-A_m) x\|\le  \frac{\sqrt {2R}}{\sqrt {\la_{m-l+1}}}  \|x\|_C,
\end{equation}
for $k>m>l$,  where $A_q=P_q AP_q$. In fact, say, dealing with the second estimate,
\[
\|(A-A_m)x\|^2 =\|(A-A_m)(1-P_{m-l})x\|^2 \le 2R \|x\|_C \, \|(1-P_{m-l})x\|
\le \frac{2R}{\la_{m-l+1}}  \|x\|_C^2.
\]

Applying this result to the operators $L$ and $H$ that satisfy 
\eqref{eqthMor4} with $R=\max (K',\sqrt K)$,  writing 
\[
L_k^*L_k-L_m^*L_m=(L_k^*-L_m^*)L_k+L_m^*(L_k-L_m),  
\]
and assuming $k>m$ for definiteness, 
we derive from \eqref{eqthMor3} (taking into account \eqref{eqthMor2}) that, for $m$ so large that $\la_{m-l+1}>1$, 
\[
\E \|\chi_k(t)-\chi_m(t)\|^2\le \|(P_k-P_m)\chi_0\|^2 
+\frac{R}{\sqrt {\la_{m-l+1}}}  e^{\al t} [\|\chi_0\|^2_C+\al t(\|\chi_0\|^2+\be)].
\]
Consequently, the sequence of continuous $\HC$-valued processes $\chi_m(t)$ is Cauchy and hence it converges to some continuous process $\chi(t)$. Since all $\chi_m(t)$ are conservative, it follows that $\chi(t)$ is conservative implying that $\|\chi(t)\|^2$ is a positive martingale. Uniform estimate \eqref{eqthMor2} implies the corresponding estimate \eqref{eqthMor1} for the limiting curve $\chi(t)$. In fact, it implies that 
$CP_k\chi_m(t)\to CP_k \chi(t)$, as $m\to \infty$ and any $k$, and hence $\|P_kC\chi(t)\|$ satisfies \eqref{eqthMor1}, and thus $\|C\chi(t)\|$ has the same bound.

It remains to prove that $\chi(t)$ satisfies equation  \eqref{eqqufiBlinint}. To this end we rewrite \eqref{eqMRappr2} 
(with our usual reduced way of writing with 
the appropriate summation assumed) as
\begin{equation}
\label{eqMRappr2a}
\chi_m(t)=P_m\chi_0 +\int_0^t (-iH-\frac12 L^* L) \chi_m(s)\, ds 
+ \int_0^t L \chi_m(s) \, dY(s)+I_m(t),
\end{equation}
with 
\[
I_m(t)=\int_0^t [-i(H_m-H)-\frac12 (L_m^*L_m-L^* L)] \chi_m(s)\,ds 
+ \int_0^t (L_m-L) \chi_m(s) \, dY(s).
\]
Then $I_m(t)$ tends to zero in the mean square sense, by \eqref{eqthMor4}. 
Thus passing to the limit in \eqref{eqMRappr2a} we get \eqref{eqqufiBlinint}.
\end{proof}

Paper \cite{MoraRebolin} presents two classes of examples satisfying assumptions of Theorem 
 \ref{thMorRebLin}. We shall give an essentially extended version of their first example supplying also more direct proofs of all conditions required.

 To this end we develop an easy verified criterion 
 for the condition of $C$-dissipativity.

\begin{prop}
\label{critdiss}
Suppose the estimates 
\begin{equation}
\label{eqcritdis}
\|[C,H]x\|\le \frac{\al}{4}\|x\|_C, \quad 
\sum_k |([L^k,C^2]\}x, L^kx)|\|\le \frac{\al}{2} \|x\|_C^2
\end{equation}
hold. Then both \eqref{MR1a} and \eqref{MR2} hold with $\be=0$.
\end{prop}

\begin{proof}
We have that 
\[
-2 Re \, (Cx,C(iH)x)-Re \, (Cx,C L^*L x)+ \|CL x\|^2 
\]
\[
=-2 \, Re \, (Cx,[C,iH]x)-Re \, (C^2Lx,L x)
-Re \, ([L,C^2]x,L x) + \|CL x\|^2 
\]
\[
=-2 \, Re \, (Cx,[C,iH]x)- Re \,  ([L,C^2]x,Lx). 
\]
This implies \eqref{MR1a}. Similarly, for $x\in \HC_m$, we have 
\[
-2 Re \, (Cx,CP_m(iH)x)-Re \, (Cx,C P_m L^*P_mL x)+\|CP_mLx\|^2 
\]
\[
=-2 \, Re \, (P_mCx,[C,iH]x)
-Re \, (P_m[L,C^2]x,P_mL x)  
\]
implying \eqref{MR2}.
\end{proof}

We can now apply this criterion to the continuous 
observation of momentum and position of a standard quantum system, where  
\begin{equation}
\label{eqgenHam}
H=\left(\frac{h}{i}\frac{\pa}{\pa x}-A(x)\right)^2+V(x)
\end{equation}
is the standard Hamiltonian of quantum mechanics (here $V(x)$ and $A(x)$ represent scalar and magnetic potentials) in $L^2(\R^d)$ and $L$ is either a position or a momentum operator, or their linear combination with constant coefficients:
\[
L=a \hat x+b \hat p, \quad \hat p= -i \frac{\pa}{\pa x},
\]
with real constants $a,b$, where $\hat x$ is the operator of multiplication by $x$.
As a direct consequence of Proposition \ref{critdiss} we can conclude the following. 

\begin{prop}
\label{critdissexam}
The conditions 
of Theorem  \ref{thMorRebLin} hold in this situation whenever $V$ and $A$ are regular enough (for instance, they are bounded with bounded first and second order derivatives), where $C$ can be chosen either as the Hamiltonian of quantum oscillator
\[
C=\hat x^2+\hat p^2,
\]
or any its power $C^q$ with any natural $q$.  Choosing $q\ge 2$ we can ensure the validity of the additional Hypothesis A with $l=1$.
\end{prop}

\begin{remark}
In \cite{MoraRebolin} this result was proved by rather lengthy calculations with Hermite polynomials only for vanishing $V,A$ and in one-dimensional case.
\end{remark}

\subsection{Linear equation: explicit calculations}

Let us now outline an alternative, more constructive, approach (developed essentially in \cite{BelKol91} and \cite{Kolok95}) to the equations arising at the continuous observation of momentum and position of a standard quantum system considered above in Proposition \ref{critdissexam}. 

 In this approach one solves explicitly the corresponding filtering equation with quadratic $V$ and linear $A$, and more general (sufficiently regular) potentials are treated via perturbation technique. Let us consider here only the case with position measurement, because including momentum just makes the calculations more lengthy (not requiring any additional ideas). For vanishing potential the corresponding filtering equation \eqref{eqqufiBlins} takes the form

\begin{equation}
\label{eqqufiBlinGa}
d\chi(t) =\frac12 (ih\De-\al^2 x^2) \chi(t) \,dt+\al x \chi(t) dY(t),
\end{equation}
where $\chi(t)\in L^2(\R^d)$, $Y(t)=(Y_1, \cdots, Y_d)(t)$ is a $d$-dimensional Brownian motion (BM), $\al$ real and $h$ positive constants.

It was proved in \cite{Kolok95} (and can be checked by direct calculations) that this equation has the Green function (fundamental solution for the Cauchy problem)
expressed in the explicit Gaussian form: 

\begin{equation}
\label{eqqufiBlinGaSo}
u_G(t,x,y)=\exp\{-\frac{\om}{2} (x^2+y^2) +\be xy -a x -b y -\ga \}
\end{equation}
where
\[
\om=\si \coth (\si Gt), \quad
\be=\si (\sinh(\si G t))^{-1}, \quad C=\sqrt{\be/(2\pi)}.
\]
\[
a= \al (\sinh(\si G t))^{-1}\int_0^t  \sinh(\si G s) dB(s)
\]
\[
b=\si G \int_0^t \frac{a(s)}{\sinh(\si G s)} ds,
\quad \ga=\frac12 \int_0^t a^2(s)ds
\]
with
$\si=\sqrt{2\al^2/ih}=\sqrt{2\al^2/h} \exp\{-i\pi/4\}$.

It follows, that for small $t$,
\[
\om=\frac{1}{iht} +\frac23 \al^2 t +O(t^3), \quad
\be=\frac{1}{iht}-\frac13 \al^2 +O(t^3).
\]
and 

\begin{equation}
\label{asympstochosc}
a \sim \frac{\al}{t} \xi(t), \quad \xi(t)=\int_0^t s \, dB(s), \quad
b \sim \al \int_0^t \frac{\xi(s)}{s^2} ds.
\end{equation}

Let us introduce
the Hilbert space $L^2_R=L^2_R(\R^d)$ of functions from $L^2(\R^d)$ with 
a finite norm squared
\[
\|f\|^2_R =\|f\|^2+\sum_j\|x_jf\|^2+\sum_j \|\frac{\pa}{\pa x_j} f\|^2.
\]

\begin{prop}
\label{propsmoo}
The resolving operator $U_t$ for the Cauchy problem to equation \eqref{eqqufiBlinGa} with the integral kernel \eqref{eqqufiBlinGaSo} takes $L^2(\R^d)$ to $L^2_R(\R^d)$ for any $t>0$ and moreover, for $t\in [0,T]$ with any $T$, and any $p\in [1,2)$
\begin{equation}
\label{eqpropsmoo}
\E \|U_t\|^p \le C, \quad  \E \|U_t\|^p_{L^2_R(\R^d)} \le C,   
\end{equation}
with a constant $C$ depending on $T$ and $p$. The norm squared $\|U_t\|^2$ has no finite expectation.   
\end{prop}

\begin{remark}
Note that neither of these estimates follow from the general Theorem
\ref{thMorRebLin}. For instance, the conservation property of the norm squared of the solution from this theorem implies only that $\E \|U_t\|^2\ge 1$.  The fact that $p=2$ is
the exact boundary between the existence and nonexistence of the expectation of 
$\|U_t\|^p$ comes out as the result of some not obvious explicit calculations. It leads to the natural questions: (i) what is the intuition behind this fact? (ii) Is it a casual fact about a particular equation, or is it a performance of some general  effect?
\end{remark}


\begin{proof}
Step 1. The first statement is evident from the form 
of the Green function.

Let us prove that the second estimate in \eqref{eqpropsmoo}
follows from the first one. To this end we observe that 
\[
\frac{\pa}{\pa x_j} u_G(t,x,y)=-\om x +\be y -a,
\quad 
\frac{\pa}{\pa y_j} u_G(t,x,y)=-\om y +\be x -b,
\]
so that 
\[
\frac{\pa}{\pa x_j} u_G(t,x,y)+\frac{\pa}{\pa y_j} u_G(t,x,y)
=(\be-\om)(y+x) -a-b.
\]

Hence, 
\[
\int x_j u_G(t,x,y) f(y) dy 
=\frac{1}{\be} \int \frac{\pa}{\pa y_j} u_G(t,x,y) f(y) dy 
+\int u_G(t,x,y)\frac{\om y+b}{\be} f(y) dy
\]
\[
=-\frac{1}{\be} \int u_G(t,x,y) \frac{\pa}{\pa y_j} f(y) dy 
+\int u_G(t,x,y)\frac{\om y+b}{\be} f(y) dy
\]
And thus
\[
\|\hat x_j U_t f\|\le \frac{1}{|\be|} \|U_t \| \, \|p_j f\|
+\frac{|\om|}{|\be|} \|U_t\| \, \|x_j f\|+\frac{|b|}{|\be|} \|U_tf\|.
\]

Similarly 
\[
\int \frac{\pa}{\pa x_j} u_G(t,x,y) f(y) dy 
=-\int \frac{\pa}{\pa y_j} u_G(t,x,y) f(y) dy 
+\int [(\be -\om) (y+x)-(a+b)]u_Gt,(x,y) f(y) dy,
\]
 so that
 \[
  \|\hat p_j U_t f\| \le \|U_t\|\, \|p_j f\|  
+(\be-\om) \|U_t\| \, \|\hat x_j f\|  
+  |a+b| \|U_t f\| + (\be-\om)  \| x_j U_t f(y) \|.
 \]
 Therefore,
\begin{equation}
\label{eqnormregU}
\E \|\hat x_j U_t f\|\le  \left( \frac{1+|\om|}{|\be|} \E \|U_t \| 
+\E \frac{|b|}{|\be|} \|U_t f\|\right) \|f\|_{L^2_R(\R^d)},
\end{equation} 
\begin{equation}
\E \|\hat p_j U_t f\| \le \left((1+\be-\om) \E \|U_t\|   
+ \E |a+b| \|U_t \| \right) \|f\|_{L^2_R(\R^d)}
+ (\be-\om) \E \| x_j U_t f(y) \|.
\end{equation}
Taking into account that $\be-\om=O(t)$, $|\be|^{-1}=O(t)$ and that $a,b$ are small for small $t$, we can in fact conclude that 
 the second estimate in \eqref{eqpropsmoo}
follows from the first one.

Step. 2.
The norm squared of $U_t$ is the norm of the operator $U_tU_t^*$ with the kernel
\[
|C|^{2m} (\pi/ Re \, \om)^{m/2}
\exp\bigl\{-\frac12 (Re \,\om -\frac{Re (\be^2)}{2\, Re \, \om})(x^2+y^2)
+\frac{|\be_G|^2}{2\, Re \, \om_G} xy
\]
\[
-(x+y)(a_R+\frac{\be_R b_R}{\om_R})
-(2\ga_R -\frac{b_R^2}{\om_R}) \bigr\}.
\]
By shifting $x,y$ on 
\[
\xi=\frac{\om_R(a_R+\frac{\be_R b_R}{\om_R})}{\om_R^2-\be_R^2},
\]
we find that the norm of this operator is the same as that of
the operator
\[
|C|^{2m} (\pi/ Re \, \om)^{m/2}
\exp\left\{-\frac12 (Re \,\om -\frac{Re (\be^2)}{2\, Re \, \om})(x^2+y^2)
+\frac{|\be|^2}{2\, Re \, \om} xy \right\}
\]
\[
\times \exp\left\{\frac{\om_R(a_R+\frac{\be_R b_R}{\om_R})^2}{\om_R^2-\be_R^2}
-(2\ga_R -\frac{b_R^2}{\om_R}) \right\}.
\]
Comparing with the resolving operator of quantum oscillator and using small time asymptotics of $\om,\be$ we find that the deterministic part of the kernel gives the norm $1+O(t)$, so that 
\[
\E \|U_t\|^2=(1+O(t))
\E \exp\left\{\frac{\om_R(a_R+\frac{\be_R b_R}{\om_R})^2}{\om_R^2-\be_R^2}
-(2\ga_R -\frac{b_R^2}{\om_R}) \right\}.
\]

Since $\ga_R$ is seen to be small, we really need to  estimate the expectation 
\[
\E \exp\left\{\frac{\om_R(a_R+\frac{\be_R b_R}{\om_R})^2}{\om_R^2-\be_R^2}
+\frac{b_R^2}{\om_R} \right\},
\]
which up to terms of lower order rewrites as 
\begin{equation}
\E \exp\left\{ \frac{\om_R}{\om_R^2-\be_R^2}(a_R^2-a_R b_R+b_R^2)\right\}
\sim \E \exp\left\{\frac{2}{\al^2 t}(a_R^2-a_Rb_R+b_R^2)\right\}.
\end{equation}

Similarly, for any $p\ge 1$,
\begin{equation}
\E \, \|U_t\|^p 
\sim \E \exp\left\{\frac{p}{\al^2 t}(a_R^2-a_Rb_R+b_R^2)\right\}.
\end{equation}

Step 3.
By \eqref{asympstochosc}, we see that asymptotically (for small $t$)
\[
Var (a_R) =Var (b_R)=2 \, Cov \, (a_R,b_R)=\al^2 t/3.
\]
Consequently,
\[
a_R^2-a_Rb_R+b_R^2=(a_R-\frac12 b_R)^2+\frac34 b_R^2,
\]
with
\[
Var (a_R-\frac12 b_R)=Var (\frac{\sqrt 3}{4}b_R)=\frac14 \al^2 t,
\quad Cov \,  (a_R-\frac12 b_R,\frac{\sqrt 3}{4}b_R)=0.
\]
Consequently, up to a multiplier of type $(1+O(t))$ we see that 
\[
\E \, \|U_t\|^p=\E \exp\{ \frac{p}{4} (z_1^2+z_2^2)\}, 
\]
with $z_1,z_2$ independent standard normal, so that $z_1^2+z_2^2$ has the
standard $\chi^2$ distribution of degree 2. This distribution has the MGF 
 $(1-2t)^{-1}$ for $t<1/2$. Therefore, $\E \,\|U_t\|^p$ is finite exactly when $p<2$.
\end{proof}

This Proposition implies the following result.
\begin{prop}
\label{propsmoo1}
There exists a unique strong solution
to \eqref{eqqufiBlinGa} for any initial condition (not necessary regular as in Theorem \ref{thMorRebLin}), the equation being satisfied generally for all $t>0$, and, for the initial condition from $L^2_R(\R^d)$, for all $t\ge 0$.
\end{prop}

Via the standard perturbation argument this result can be extended to more general equations 
\begin{equation}
\label{eqqufiBlinGaM}
d\chi(t) =(-iH -\frac12 \al^2 x^2) \chi(t) \,dt+\al x \chi(t) dY(t),
\end{equation}
with $H$ of form \eqref{eqgenHam} with sufficiently regular $V$ and $A$. For instance, the following result is straightforward.

\begin{prop}
\label{propsmoo2}
If $A=0$ and $V$ is bounded with bounded continuous first and second order derivatives,
there exists a unique strong solution
to \eqref{eqqufiBlinGaM} for any initial condition, the equation being satisfied generally for all $t>0$, and, for the initial condition from $L^2_R(\R^d)$, for all $t\ge 0$.
\end{prop}

\subsection{Nonlinear equation: weak solutions}
\label{secpureeqnonlin}

In the setting of Theorem \ref{thMorRebLin}, solutions to the corresponding nonlinear equation have only been built in the weak (probabilistic) sense.
Namely, a $C$-weak solution of equation \eqref{eqqufiBnonlin} with initial condition $\phi_0$ on a time interval $[0,T]$ is the pair
$\{\phi(t), B(t)\}$ of adapted continuous processes defined on some filtered probability space $(\Om, \FC, \FC_t,\Q)$ (satisfying the usual conditions), where $\phi(t)$ is $\HC$-valued, $B(t)$ is an $n$-dimensional standard BM, s.t. (i) $\phi(0)=\phi_0$ and $\|\phi(t)\|=1$ for all $t$ a.s., (ii) $\phi(t)\in Dom(C)$ a.s. and
\[
\sup_{t\in [0,T]} \E \| C \phi(t)\|^2 <\infty,
\]
(iii) equation \eqref{eqqufiBnonlin} holds. This solution is $C$-strong if $\phi(t)$ is measurable with respect to the augmented $\si$-algebra generated by the BM $B(t)$. 

The following result is a simplified version of Theorem 1 from \cite{MoraRebo}:

\begin{theorem} (Mora-Rebolledo) Under Hypothesis MR (with either \eqref{MR1} or \eqref{MR2}), for any $\phi_0\in Dom(C)$ with $\|\phi_0\|=1$ there exists a unique in law (that is, all solutions have identical finite-dimensional distributions) 
$C$-weak solution $\{\phi(t), B(t)\}$ to equation \eqref{eqqufiBnonlin}, and moreover, this solution satisfies the estimate
\begin{equation}
\label{eqthMornon1}
\E \|C\phi(t)\|^2 \le e^{\al t} [\|C\phi_0\|^2+\al t(\|\phi_0\|^2+\be)].
\end{equation}
\label{thMorRebNonlin}
\end{theorem}

\begin{proof}
 We sketch here a simplified argument. First of all, the existence of a solution is mostly straightforward from Theorem \ref{thMorRebLin} and the remark above (see \eqref{eqdefinnov}) that, if $\chi(t)$ solves \eqref{eqqufiBlin}, then $\phi(t)=\chi(t)/\|\chi(t)\|$ solves \eqref{eqqufiBnonlin}. The estimate \eqref{eqthMornon1} follows from \eqref{chisqdensity} and
\eqref{eqthMor1}. The difficult part of the proof of \cite{MoraRebo} concerns uniqueness. It is proved there independently of the link with the linear equation. However, this link can be used to obtain a simpler proof, as was noted in \cite{K}.
Namely: let $\phi(t)$ have unit norms for all $t$ and satisfy the
nonlinear equation \eqref{eqqufiBnonlin}. Define $\|\chi(t)\|^{-2}$ as the solution
(with the initial condition equal to $1$) to the equation
 \begin{equation}
\label{chisquare4}
d\frac{1}{\|\chi(t)\|^2}=-\frac{2}{\|\chi(t)\|^2}
(\phi(t),L_S \phi(t)) dB(t).
\end{equation}
 Then the vectors
$\chi(t)=\phi(t) \|\chi(t)\|$ satisfy the linear equation \eqref{eqqufiBlin} with $Y(t)$ given by \eqref{eqdefinnov}.
Consequently, to each solution of the nonlinear equation there corresponds a 
(uniquely defined) solution of the linear one, and vise versa. Since we have uniqueness for the latter, we derive uniqueness for the former. 
\end{proof}

\subsection{Nonlinear equation: strong solutions}
\label{secpureeqnonlin2}

Proving well-posedness of equation  
\eqref{eqqufiBnonlin} in the strong probabilistic sense (that is, with $\phi(t)$ being adapted to the filtration generated by the BM $B(t)$) is of course a  more difficult task,
as it does not seem to be derivable from the linear equation. In the case of a bounded $L$ 
and an arbitrary self-adjoint operator $H$ in the Hilbert space $\HC$ the well-posedness in the strong sense of equation \eqref{eqqufiBnonlin} was proved in \cite{KolQuantLLN}. We shall show now that a modification of our proof of Theorem  \ref{thMorRebLin} above can be used to construct a strong solution to \eqref{eqqufiBnonlin}.

\begin{theorem}
\label{strongpure}
Assume that Hypothesis A holds and Hypothesis MR holds with the second alternative \eqref{MR2}.
Then, given a BM $B(t)$ and $\phi_0\in Dom(C)$ with $\|\phi_0\|=1$ there exists a unique $C$-strong solution $\phi(t)$ to equation \eqref{eqqufiBnonlin}, which satisfies estimate \eqref{eqthMornon1}.
\end{theorem}

\begin{proof}
To shorten formulas we perform our argument only for self-adjoint operators $L$.
Generally one just has to work in the same way with equation  \eqref{eqqufiBnonlin}.


Step 1.
Recall that equation \eqref{MR2} (unlike \eqref{MR1}) can be written in the form of the quantum filtering equation 
\begin{equation}
\label{eqMRappr22}
\chi_m(t)=P_m\chi_0 +\int_0^t P_m (-iH_m-\frac12 L_m^2) \chi_m(s)ds 
+\int_0^t  L_m \chi_m(s) \, dY(s),
\end{equation}
with the operators $H_m=P_mHP_m$ and $L_m=P_mLP_m$ instead of $H$ and $L$. 
The corresponding nonlinear quantum filtering equation takes the form  
\begin{equation}
\label{eqMRapprnonl}
\phi_m(t)=P_m\phi_0 
+\int_0^t [-iH_m -\frac12 (L_m-\langle L_m\rangle_{\phi_m(s)})^2] \phi_m(s)ds 
+\int_0^t (L_m-\langle L_m \rangle_{\phi_m(s)}) \phi_m(s) \, dB(s).
\end{equation}

This finite-dimensional situation is well understood. Assuming $\chi_0=\phi_0$ 
with $\|\phi_0\|=1$, we know that equation 
\eqref{eqMRapprnonl} has the unique strong solution $\phi_m(t)$ on any interval $[0,T]$
such that $\|\phi(t)\|=1$ for all $t$. It is also known that (i) the process $\|\chi_m(t)\|^{-2}$ given by \eqref{chisquare4} is a positive martingale, (2) the processes $Y^m(t)=B(t)-2\int_0^t \langle L\rangle_{\phi_m(s)}ds$ 
and $\chi_m(t) =\phi_m(t) \|\chi_m(t)\|$ satisfy
equation \eqref{eqMRappr22}. Assuming that the BM $B(t)$ is defined on some stochastic 
basis $(\Om, \FC, \FC_t, \Q)$ with $\FC_t$ being the (augmented) filtration generated by $B$, it will follow (by Girsanov's theorem) that $Y^m(t)$ is a BM on the basis $(\om, \FC, \FC_t, \P_m)$, where the measure $P_m$ has the density $\|\chi_m\|^{-2}$ with respect to
$\Q$. We shall write $\E=\E_{\Q}$ for the expectation with respect to $\Q$.

Applying Girsanov's theorem and \eqref{eqthMor1}
for $\chi_m$, $Y_m$ and measure $\P_m$, yields the following bounds 
for the solutions $\phi_m(t)$:
\begin{equation}
\label{eqstrongpure1}
\E_{\Q} \,  \|C\phi(t)\|^2 \le e^{\al t} [\|C\phi_0\|^2+\al t(1+\be)].
\end{equation}

Step 2.
Following the line of arguments used for the linear equation, we aim to show that the sequence $\phi_m(t)$ converges to a solution of equation \eqref{eqqufiBnonlins}.

Subtracting equations for $k$ and $m$ we obtain (omitting argument $t$ for brevity)
\[
d(\phi_k-\phi_m)=-[i(H_k-H_m) +\frac12  (L_k^2-L_m^2)]\phi_k\, dt
+(L_k-L_m)\phi_k \, dB(t) 
\]
\[
-[i H_m +\frac12 L_m^2](\phi_k-\phi_m)\, dt
+L_m (\phi_k-\phi_m) \, dB(t)
\]
\[
+(L_k-L_m) \langle L_k\rangle_{\phi_k} \phi_k \, dt
+L_m ( \langle L_k\rangle_{\phi_k} -\langle L_m\rangle_{\phi_m})\phi_k \, dt
+L_m \langle L_m\rangle_{\phi_m} (\phi_k-\phi_m) \, dt
\]
\[
-\frac12 (\langle L_k\rangle_{\phi_k}^2 - \langle L_m\rangle_{\phi_m}^2)\phi_k\, dt
-\frac12 \langle L_m\rangle_{\phi_m}^2 (\phi_k-\phi_m) \, dt 
\]
\[
- (\langle L_k\rangle_{\phi_k} - \langle L_m\rangle_{\phi_m})\phi_k\, dB(t)
- \langle L_m\rangle_{\phi_m} (\phi_k-\phi_m) \, dB(t). 
\]

And therefore
\[
d \, \E \, (\phi_k-\phi_m)^2
=-2 \, \E \, Re \, (\phi_k-\phi_m, i(H_k-H_m) \phi_k) \, dt 
\]
\[
+\E \, Re \,  (\phi_k-\phi_m, (L_k^2-L_m^2)\phi_k) \, dt
+\E \, \|(L_k-L_m)\phi_k\|^2 \, dt 
\]
\begin{equation}
\label{eqthMor3a}
+2\, \E \, Re \, ((L_k-L_m) \phi_k, L_m (\phi_k-\phi_m)) \, dt 
\end{equation}
\[
+2 \, \E \, Re \left((L_k-L_m) \langle L_k\rangle_{\phi_k} \phi_k 
+L_m ( \langle L_k\rangle_{\phi_k} -\langle L_m\rangle_{\phi_m})\phi_k 
+L_m \langle L_m\rangle_{\phi_m} (\phi_k-\phi_m), \phi_k-\phi_m\right) \, dt
\]
\[
- \E \, Re \, \left((\langle L_k\rangle_{\phi_k}^2 - \langle L_m\rangle_{\phi_m}^2)\phi_k, \phi_k-\phi_m\right) \, dt 
+ \E \, \|(\langle L_k\rangle_{\phi_k} - \langle L_m\rangle_{\phi_m})\phi_k \|^2 \, dt. 
\]

The first four terms on the right hand side are the same as in linear case and are dealt with as in the liner case:  
they are bounded by terms tending to zero, as  $m,k \to \infty$. 
The same concerns other terms containing $L_k-L_m$, in particular, those arising
 by writing 
\[
\langle L_k\rangle_{\phi_k}-\langle L_m\rangle_{\phi_m}
=(L_k \phi_k, \phi_k-\phi_m)+(L_k (\phi_k-\phi_m),\phi_m)
+((L_k-L_m)\phi_m,\phi_m).
\]
However further we meet difficulties not present in the linear case: the terms of type 
$(L_m(\phi_k-\phi_m), \phi_k-\phi_m)$. In case of bounded $L$, these terms would be bounded 
by $\|\phi_k-\phi_m\|^2$ and direct application of Gronwall's inequality would complete a proof that the sequence $\phi_k$ is Cauchy. Here we need an additional argument.

Step 3. Let us introduce the stopping times 
\[
\tau_n=\inf \{t\le T: \sup_k \|\phi_k(t)\|_C \ge n\}
\]
(or $\tau_n=T$, if $ \sup_k \|\phi_k(t)\|_C < n$ for all $t\le T$).
Then we get from \eqref{eqthMor3a} and \eqref{eqthMor4} 
(and using Doob's optional sampling theorem) that 
\[
\E \, \|\phi_k-\phi_m\|^2(t\wedge \tau_n) 
\le R \bigl( \|(P_k-P_m)\phi(0)\|^2 +\frac{1}{\la_{m-l+1}}
+n^2 \int_0^t  \E \, \|\phi_k-\phi_m\|^2(s\wedge \tau_n) \, ds \bigr),
\]
with a constant $R$ (depending on $T$). Consequently, by Gronwall's lemma,
\[
\E \, \|\phi_k-\phi_m\|^2(t\wedge \tau_n) 
\le R \, e^{tRn^2} \bigl( \|(P_k-P_m)\phi(0)\|^2 +\frac{1}{\la_{m-l+1}}\bigr).
\]
Therefore the processes $\phi_k(t\wedge n)$ have a limit in the mean square sense. Since these limits are clearly consistent, we can denote all these limits by a single letter $\phi(t\wedge \tau_n)$ and to conclude that these processes satisfy the integral version of equation \eqref{eqqufiBnonlin}, where all upper bounds for the integrals are taken as $t\wedge \tau_n$ instead of $t$.

It remains to observe that $\lim_{n\to \infty}\tau_n=T$ a.s. In fact, otherwise, we would have a set of positive measure, where $\tau_n \ge K$ for all $n$ with some $K<T$, so that there $\sup_k \max_t \|\phi_k(t)\|_C=\infty$. Consequently, there would exist $k$ such that 
$\max_t \|\phi_k(t)\|_C=\infty$, which contradicts the continuity and boundedness of $\phi_k$.

Therefore, passing to the limit $n\to \infty$ we can conclude that $\phi(t)=\lim_{n\to\infty} \phi(t\wedge \tau_n)$ satisfies \eqref{eqqufiBnonlin}. 

Step 4. Finally, the uniqueness of solution is as easy as in the linear case. One writes down the equation for $\E_{\Q} \, \|\phi^1(t)-\phi^2(t)\|^2$ of two solutions with the same initial condition, which is similar to \eqref{eqthMor3a}, but is much simpler, 
because $H,L$ are fixed:
\[
\frac{d}{dt} \, \E \, 
\|\phi^1-\phi^2\|^2
= 2 \, \E \, Re \left( 
( \langle L\rangle_{\phi^1} -\langle L\rangle_{\phi^2}) L\phi^1 
+ \langle L\rangle_{\phi^2} L (\phi^1-\phi^2), \phi^1-\phi^2\right)
\]
\[
- \E \, Re \, \left((\langle L\rangle_{\phi^1}^2 - \langle L\rangle_{\phi^2}^2)\phi^1, \phi^1-\phi^2\right) 
+ \E \, \|(\langle L\rangle_{\phi^1} - \langle L\rangle_{\phi^2})\phi^1 \|^2 . 
\]

As above, we introduce stopping times 
\[
\tau_n=\min\{t\le T: \max(\|\phi^1(t)\|_C, \|\phi^2(t)\|_C)\ge n\}.
\]
Applying Gronwall's lemma we then conclude that 
$\phi^1(t\wedge \tau_n)=\phi^2(t\wedge \tau_n)$. Passing to the limit $n\to \infty$  
yields the equation $\phi^1(t)=\phi^2(t)$ for all $t$, as required.
\end{proof}

\section{Stochastic Lindblad (or quantum master) equations: linear version}
\label{seclineq}

The first mathematical problem with equation \eqref{Lindstoch},
that is, equation 
\[
d\ga(t)=-i[H,\ga(t)] \, dt +\LC_L \ga(t) \, dt +(L\ga(t)+\ga(t) L^*) dY(t),
\]
with
\[
\LC_L\ga =L\ga L^*-\frac12 L^*L\ga -\frac12 \ga L^*L
=L\ga L^*-\frac12 \{L^*L,\ga\},
\]
is that it is not obvious, in which space it should be considered. 
Since we are interested in trace-class operators, the most natural
space from physical point of view  would be the Banach space $\HC^1_S$ of self-adjoint
trace-class operators in $\HC$. However, the classes of Banach spaces,
for which a satisfactory extension of Ito stochastic calculus was developed,
namely the so-called UMD spaces, spaces of martingale type 2 and spaces with a smooth norm
(see review \cite{VanNeerven}) do not include $\HC^1$. 
In \cite{K} in was suggested to work with this equation in a larger Hilbert space 
$\HC^2$ of Hilbert-Schmidt operators or more precisely in its closed subspace $\HC^2_S$ of self-adjoint operators.  We shall follow this idea here, having in mind the additional difficulty arising for unbounded $L$: solution-processes may not be square integrable in the Hilbert-Schmidt class.
In fact, even for factorized solutions of type $\ga(t)=\chi(t)\otimes \bar \chi(t)$ with $\chi(t)$ solving the pure state linear filtering equation,
$\E \, {\tr} \, \ga(t)^2=\E \|\chi(t)\|^4$, which may not be finite generally.


A natural class for well-posedness for  equation \eqref{Lindstoch} in the setting of Hypothesis MR is the cone $\TC_C^{S+}(\HC)$ introduced in 
\cite{ChebQuez}: $\TC^{S+}_C(\HC)$ is the cone of trace-class positive operators $\ga$ in $\HC$ such that one of three equivalent properties holds: (i) $C \sqrt \ga \in \HC^2$ (including the statement that the image of $\sqrt \ga$ belongs to the domain of $C$), (ii) $ \sqrt \ga C\in \HC^2$,
(iii) $C\ga C$ (more precisely, its unique extension by continuity) is of trace class. The easiest way to see that these properties are equivalent
is to use the spectral basis of $C$, where it is diagonal with eigenvalues $\la_i$ on  the diagonal, since then
\[
{\tr} \, C\ga C=\sum_j \la_j^2 \ga_{jj}
=\sum_{j,k}\la_j^2 |\sqrt \ga_{jk}|^2
=\|C\sqrt \ga\|_{\HC^2} =\|\sqrt \ga C\|_{\HC^2}.
\]
The closed linear span $\TC_C^S(\HC)$ of the cone $\TC_C^{S+}(\HC)$ is a Banach space when equipped with the norm $\|\ga\|_C={\tr} \, (C |\ga|C)$. 

We shall call a continuous $\HC^2$-valued process $\ga(t)$ a $C$-strong solution of \eqref{Lindstoch}, if it is adapted to the filtration generated by $Y(t)$, satisfies \eqref{Lindstoch} in $\HC^2$  and enjoys the bound  
\[
\sup_{t\in [0,T]} \E \, \|\ga(t)\|_C  <\infty.
\]

Let us say that such a solution with an initial condition $\ga(0)$ {\it respects finite-dimensional approximations}
if it is an $L^2$-limit (in the space of $\HC^2$-valued random variables),
as $m\to \infty$, of finite dimensional approximations $\ga_m(t)$ defined as the solutions to the equations of type  \eqref{Lindstoch} in $\HC_m$, where the operators 
$L_m=P_mL P_m$ and $H_m=P_mHP_m$ are taken instead of $L$ and $H$, 
\begin{equation}
\label{Lindstochapp}
d\ga_m(t)=-i[H_m,\ga_m(t)] \, dt +\LC_{L_m} \ga_m(t) \, dt 
+(L_m\ga(t)+\ga(t) L_m^*) \, dY(t),
\end{equation}
considered with the initial condition $\ga_m(0)=P_m\ga(0)P_m$.

\begin{remark}
For the stochastic differential (and its integral) on the r.h.s. of \eqref{Lindstoch} to be well-defined it is needed that the integral
\begin{equation}
\label{eqItoNeed}    
\int_0^t \|L\ga(t)+\ga(t) L^*\|^2_{\HC^2} \, dt
= \int_0^t {\tr} (L\ga(t)+\ga(t) L^*)^2 \, dt
\end{equation}
is finite a.s. One can see that this condition is fulfilled under the assumptions of the definition of $C$-strong solution. In fact,
under these assumptions $L\ga(t)$ is an a.s. continuous curve both in $\HC^2$ and $\HC^1_S$ (because $\ga(t)$ is continuous and $L$ is a continuous operator from  $\TC_C^{S+}(\HC)$ to $\HC^1$, which holds by the closed graph theorem), hence its square norm is continuous a.s. and therefore integrable. For unbounded $L$ one can hardly expect to have a finite expectation of expression \eqref{eqItoNeed}, which complicates the arguments of analysis.
\end{remark}

\begin{theorem}
\label{LindStochLin1}
Under Hypothesis MR, for any $\ga_0\in \TC_C(HC)$ there 
exists a unique $C$-strong solution to equation \eqref{Lindstoch} with initial condition $\ga_0$, which (i) respects finite-dimensional approximations, (ii) is conservative in the sense that $\E \, {\tr} \, \ga(t)={\tr} \, \ga_0$ and (iii) satisfies the growth estimate 
 \begin{equation}
\label{Lindstochquadex1}
\E \|\ga(t)\|_C 
\le e^{\al t} [\|\ga_0\|_C +\al t ({\tr}\, \ga_0+\be)].
\end{equation}
Moreover, if $\ga_0$
is positive, then so is the solution $\ga(t)$, 
and ${\tr}\, \ga(t)$ is a positive martingale given by the equation
\begin{equation}
\label{Lindstochmart}
{\tr} \,\ga(t)={\tr} \, \ga_0+\int_0^t {\tr} (L\ga(s)+\ga(s) L^*) \, dY(s).
\end{equation}
Finally 
\[
\ga(t)-\int_0^t \LC_L(\ga(s))ds
\]
is a $\HC^2$- and $\HC^1$-valued martingale, and $\E \ga(t)$ solves  
the quantum stochastic master equation \eqref{eqquantmas}.
\end{theorem}

\begin{remark}
Since all reasonable solutions are built from finite-dimensional approximations, our condition that solutions respect finite dimensional approximation seems to be satisfactory
for any application. Moreover, it implies the Markov property for the solution and the semigroup (propagator) property for resolving operators. 
Aesthetically however, it would be nice to prove uniqueness without this restriction. That is, what we will do in the next theorem under additional Hypothesis A.  
\end{remark}

\begin{proof}

Step 1. 

Since finite-dimensional equations are  well-posed (see e.g. \cite{K}), it follows that
the requirement of the theorem automatically implies uniqueness. However, the convergence of finite-dimensional approximations does not carry out in direct analogy with the case of pure states.  Instead, we prove it as a consequence of the finite-dimensional case.

By linearity, for the proof of existence, it is sufficient to look at positive initial conditions $\ga_0$. Being of trace class, $\ga_0$ is also a Hilbert-Schmidt operator.
Therefore there exists a orthonormal basis $\{e_{k0}\}$ in $\HC$ 
such that $\ga_0$ can be presented as a convergent series (both in $\HC^1_S$ and $\HC^2_S$) 
 \[
 \ga_0=\sum_{k=1}^{\infty} p_k e_{k0}\otimes \bar e_{k0},
 \]
 with a summable non-negative sequence $\{p_k\}$.
 Clearly, if $C_{jk}$ are matrix elements of $C$ in the basis $\{e_k\}$, then 
 \[
 \|\ga_0\|_C=\sum_{j,k}|C_{jk}|p_k,
 \]
 so that the series converges also in $\TC^S_C(\HC)$. 
 Hence a solution to  equation \eqref{Lindstoch} can be written explicitly 
 as the monotone convergent series (in the sense of $L^1$-convergence of $\HC^1_S$-valued or  $\HC^2_S$-valued random variables) 
 \[
 \ga(t)=\sum_{k=1}^{\infty} p_k e_{k}(t)\otimes \bar e_{k}(t),
 \]
 with $e_{k}(t)$ solving the corresponding pure state
 equation according to Theorem \eqref{thMorRebLin}. 

Step 2. 

From the discussion above, it is clear that it is sufficient to prove the convergence of finite-dimensional approximations for decomposable initial conditions 
$\ga(0)=\chi(0)\otimes \bar \chi(0)$. Clearly,
$P_m \ga(0)P_m=P_m\chi \otimes \overline{P_m\chi}$. Let $\chi_m(t)$ be the solution of 
the approximate pure state filtering equations of type \eqref{eqqufiBlinint} described in Theorem \ref{thMorRebLin}. 
By this theorem, $\chi_m(t)$ converge (in the Hilbert space of square integrable 
$\HC$-valued random variables), as $m\to \infty$, to the solution   
of equation \eqref{eqqufiBlinint} with the initial condition $\chi(0)$.
Since
\[
{\tr}\, |\xi_1\otimes \bar \xi_1-\xi_2 \otimes \bar \xi_2|
\le (\|\xi_1\|+\xi_2\|)\|\xi_1-\xi_2\|
\]
for any two vectors $\xi_1, \xi_2$, it follows that the corresponding solutions 
$\ga_m(t)=\chi_m(t)\otimes \overline{\chi_m(t)}$ of equation \eqref{Lindstochapp}
converge to the solution $\ga(t)=\chi(t)\otimes \overline{\chi(t)}$ in expectation (not necessarily in the mean squared!) in either of spaces of $\HC^1$- or $\HC^2$-valued random variables.

 All required properties of solutions follow automatically from finite-dimensional approximations, where they are known to hold (and easy to obtain).
\end{proof}

\begin{theorem}
\label{LindStochLin2}
Under Hypothesis MR and A the solution $\ga(t)$ of the previous theorem is unique without the restriction of respecting finite-dimensional approximations. 
\end{theorem}

\begin{proof}

Let $\tilde \ga(t)$ be any solution to  \eqref{Lindstoch} with the initial condition $\ga(0)$. Let  us compare its approximations $\tilde \ga_m(t)=P_m\ga(t)P_m$ with the approximations $\ga_m(t)$ to the solution $\ga(t)$ built in the previous theorem.
 If we show that $\|\ga_m(t)-\tilde \ga_m(t)\| \to 0$, we are done.

We see that $\tilde \ga_m(t)$ satisfy the equations (omitting variable $t$ for brevity)
\[
d\tilde \ga_m=-i(P_mH\ga P_m-P_m\ga H P_m) \, dt 
+(P_mL \ga L^*P_m-\frac12 P_m L^*L \ga P_m-\frac12 P_m  \ga L^*L) \, dt 
\]
\[
+(P_mL\ga P_m+P_m\ga L^* P_m) \, dY(t).
\]
Writing 
\[
P_mH\ga P_m-P_m\ga H P_m
=H_m \tilde \ga_m -\tilde \ga_m H_m 
+P_m H(1-P_m)\ga P_m- P_m\ga (1-P_m)H P_m,
\]
and 
\[
P_mL\ga L^*P_m=L_m\tilde \ga_m L^*_m
\]
\[
+P_mL(1-P_m)\ga P_m L^*P_m 
+P_m L P_m \ga (1-P_m) L^* P_m+P_m L(1-P_m)\ga (1-P_m) L^*P_m,
\]
\[
 P_m L^*L \ga P_m+P_m  \ga L^*L P_m
 =L_m^*L_m \tilde \ga_m +\tilde \ga_m L_m^*L_m 
\]
\[
+P_mL^*(1-P_m) LP_m \ga P_m +P_m L^* P_m L(1-P_m) \ga P_m 
+P_m L^*(1-P_m)L (1-P_m) \ga P_m 
\]
\[
+P_m \ga (1-P_m) L^* P_m L P_m +P_m \ga P_m L^*(1-P_m) L P_m 
+P_m \ga (1-P_m)L^* (1-P_m) L P_m ,
\]
we see that $\tilde \ga_m$ satisfies equation \eqref{Lindstochapp} up to an additive correction $I=I_1+I_2+I_3$ with $I_1(0)=I_2(0)=I_3(0)=0$ and
\[
dI_1=-i\bigl(P_m H(1-P_m)\ga P_m- P_m\ga (1-P_m)H P_m\bigr) \, dt,
\]
\[
dI_2=\bigl(P_mL(1-P_m)\ga P_m L^*P_m 
+P_m L P_m \ga (1-P_m) L^* P_m+P_m L(1-P_m)\ga (1-P_m) L^*P_m\bigr)\, dt
\]
\[
-\frac12 \bigl( P_mL^*(1-P_m) LP_m \ga P_m +P_m L^* P_m L(1-P_m) \ga P_m 
+P_m L^*(1-P_m)L (1-P_m) \ga P_m\bigr) \, dt 
\]
\[
-\frac12 \bigl( P_m \ga (1-P_m) L^* P_m L P_m +P_m \ga P_m L^*(1-P_m) L P_m 
+P_m \ga (1-P_m)L^* (1-P_m) L P_m\bigr) \, dt,
\]
\[
dI_3=\bigl(P_m L(1-P_m)\ga P_m +P_m \ga (1-P_m) L^* P_m\bigr) \, dY.
\]
Therefore $\tilde \ga_m-\ga_m$ has initial condition zero and also satisfies \eqref{Lindstochapp} up to the additive correction $I$.  

Let $\Phi_m^{t,s}\, \ga$, $t\ge s$, denote the solution to equation \eqref{Lindstochapp} with
the initial condition at time $s$. As everything is well defined in finite-dimensional setting, we can write, by the Duhamel principle that 
\[
\tilde \ga_m(t)-\ga_m(t)=\int_0^t \Phi_m^{t,s}\, dI(s) .
\]
We want to show that this integral tends to zero, as $m\to \infty$.
To this end we shall use the following consequence of estimates 
\eqref{eqthMor4}: under he assumptions of these estimates, we have
\begin{equation}
\label{eqthMor4ope}
{\tr} \, |(A-A_m) \ga|\le  \frac{2\sqrt {2R}}{\sqrt {\la_{m-l+1}}}
\, {\tr}\, (C\ga C),
\end{equation}
for any positive $\ga$. In fact, as above, any positive $\ga$ can be written as
$\ga=\sum_{k=1}^{\infty} p_k \xi_k\otimes \bar \xi_k$, so that 
\[
{\tr} \, (C\ga C)=\sum_k p_k \|C\xi_k\|^2.
\]
On the other hand, for any operator $X$ and any vector $\xi$,
\[
{\tr} \, |X(\xi\otimes \bar \xi)|={\tr} \, |(\xi\otimes \bar \xi)X|
=\|\xi\| \, \|X\xi\|.
\]
Therefore,
\[
{\tr} \, |(A-A_m) \ga| \le \sum_k p_k \, {\tr} |(A-A_m)(\xi_k\otimes \bar \xi_k)|
=\sum_k p_k \|(A-A_m)(\xi_k)\| \, \|\xi_k\| 
\]
\[
\le \sum_k p_k   \frac{\sqrt {2R}}{\sqrt {\la_{m-l+1}}}\|\xi_k\|_C \|\xi_k\|
\le  \frac{2\sqrt {2R}}{\sqrt {\la_{m-l+1}}} \|C\xi_k\|^2,
\]
implying \eqref{eqthMor4ope}.

It follows that 
\[
{\tr}\, \bigl|P_m H(1-P_m)\ga(t) P_m- P_m\ga(t) (1-P_m)H P_m\bigr|
\]
\[
\le {\tr} \, |(H-H_m)\ga(t)|+\, {\tr} \, |\ga(t) (H-H_m)|
\le  \frac{4\sqrt {2R}}{\sqrt {\la_{m-l+1}}}
\, {\tr}\, (C\ga(t) C),
\]
which tends to zero uniformly in $t\in [0,T]$ according to the definition 
of a strong $C$-solution $\ga(t)$.

Therefore,
\[
\E \, {\tr} \, \bigl| \int_0^t \Phi_m^{t,s}\, dI_1(s) \bigr| \to 0,
\]
as $m\to \infty$. 

Quite similarly one shows that 
\[
\E \, {\tr} \, \bigl| \int_0^t \Phi_m^{t,s}\, dI_2(s) \bigr| \to 0,
\]
as $m\to \infty$. 

The most difficult case is $I_3$, which requires the usual trick with stopping times. Namely,  let $\tau_n$ be the stopping time when $\|L\ga(t)\|_{\HC^1}$ 
reaches the level $n$ (alternatively, one can use as well the time when $\|L\ga(t)\|$
reaches level $n$). Then
\[
\E \, \|I_3(t\wedge \tau_n)\|^2\le 4n\int_0^t \E \|(L-L_m)\ga(s) \|_{\HC^1} \, ds 
\le \frac{\kappa n}{\sqrt{\la_{m-l+1}}} \, {\tr}\, \int_0^t (C\ga (s) C)
\]
with a constant $\kappa$. It follows that 
\[
\tilde \ga(t\wedge \tau_n)=\ga (t\wedge \tau_n),
\]
for any $n$. It remains to observe that $\tau_n\to \infty$, 
as $n\to \infty$ ($\tau_n$ becomes greater than any fixed $T$), a.s. In fact, otherwise, we would have that $\sup_{t\in [0,T]} {\tr}\, |L\ga(t)|=\infty$ on a set of positive measure, which is impossible, because the curve ${\tr} \, |L\ga(t)|$ is continuous and bounded a.s. 
 Therefore, passing to the limit $n\to \infty$, we get $\tilde \ga(t)=\ga(t)$.
 \end{proof}

 As a direct consequence we obtain the well-posedness for the master equation \eqref{eqquantmas}.

\begin{corollary}
Under Hypotheses MR and A, for any $\ga_0\in \TC_C(\HC)$ there 
exists a unique $C$-strong solution to equation \eqref{eqquantmas} with initial condition $\ga_0$. This solution is conservative (in the sense that it preserves the trace) and satisfies the growth estimate 
 \begin{equation}
\label{Lindstochquadex1cor}
 \|\ga(t)\|_C 
\le e^{\al t} [\|\ga_0\|_C +\al t ({\tr}\, \ga_0+\be)],
\end{equation}
Moreover, if $\ga_0$
is positive, then so is the solution $\ga(t)$. 
\end{corollary}

In fact, existence was obtained by passing in Theorem  \ref{LindStochLin1}, and uniqueness is obtained by exactly the same argument as in Theorem \ref{LindStochLin2}, though simpler, as the correction $I$ does not contain the most annoying term $I_3$.

This corollary is of course well known, but we include it here as a consequence of our results for completeness, because it was achieved with many efforts of different authors by constantly improved and simplified arguments, see e.g.
\cite{Cheb89}, \cite{ChebFagn}, \cite{ChebQuez}, \cite{Mora13}, \cite{MoraRebo} and references therein.

To complete this section let us formulate the result on the continuous dependence of the solutions to  equation \eqref{Lindstoch} on the Hamiltonian $H$. Its proof is a straightforward application of the perturbation technique and will be, therefore, omitted.

\begin{theorem}
\label{LindStochLin3} 
Consider two equations of type \eqref{Lindstoch} with one and the same $L$, but with two different Hamiltonian $H_1$, $H_2$, both satisfying assumptions of Theorem \ref{LindStochLin1} and such that $H_1-H_2$ is a bounded operator. Then, for the corresponding solutions $\ga_1(t)$, $\ga_2(t)$, with the same initial condition 
$\ga(0)\in \TC_C^{S+}(\HC)$,
we have the following estimate:
\begin{equation}
\label{eqLindstochLin21}
\E \, {\tr}\, |\ga_1(t)-\ga_2(t)|\le 2 t \, \|H_2-H_1\| \, {\tr} \, \ga_0,
\end{equation}
\end{theorem}

\section{Stochastic Lindblad (or quantum master) equations: normalized version}
\label{secnormeq}

$C$-solutions for the nonlinear quantum filtering equation are defined analogously to the linear case. For a BM $B(t)$ on a probability space $(\Om, \FC, \Q)$, a continuous $\HC^2$-valued process $\rho(t)$ is called a $C$-{\it strong solution} of \eqref{Lindstochnorm1} with a positive initial condition of unit trace, if it is adapted to the filtration generated by $B(t)$, satisfies \eqref{Lindstochnorm1} in $\HC^2$, preserves the trace, has unit trace for all $t$ a.s., and is uniformly bounded in $\TC_C^S(\HC)$:  
\[
\sup_{t\in [0,T]} \E_{\Q} \, \|\rho(t)\|_C  <\infty.
\]
A pair of continuous processes $(\rho(t), B(t))$, adapted to a filtration $\FC_t$ of a certain stochastic basis, is referred to as a $C$-{\it weak (probabilistically) solution} to \eqref{Lindstochnorm1} with an initial positive $\rho(0)$ of unit trace,
if $B(t)$ is a BM and $\rho(t)$ satisfies \eqref{Lindstochnorm1} in $\HC^2$, preserves the trace and is uniformly bounded in $\TC_C^S(\HC)$.  

It is well known that solutions of nonlinear equation can be built 
formally from the linear one. Namely, from \eqref{Lindstochmart} we can derive by Ito's formula that
\[
d\frac{1}{{\tr} \, \ga(t)}=-\frac{1}{({\tr} \, \ga(t))^2} \, {\tr} \, (L\ga(t) +\ga(t) L^*) dY_t
+\frac{1}{({\tr} \, \ga(t))^3} [{\tr} \, (L\ga(t)+\ga(t) L^*)]^2 \, dt.
\]
Hence by Ito's product rule we check that the normalised density operator
$\rho(t)=\ga(t)/{\tr} \, \ga(t)$ satisfies the equation
 \[
d\rho=-i[H, \rho(t)] \, dt +\LC_L \rho(t)\, dt
\]
 \begin{equation}
\label{Lindstochnorm}
+(L\rho(t)+\rho(t) L^*-\rho(t)\, {\tr} \, (L\rho(t)+\rho(t) L^*) )
[dY_t-{\tr} \, (L\rho(t)+\rho(t) L^*) dt].
\end{equation}

Therefore, in terms of the {\it innovation process}
 \begin{equation}
\label{outputinnovation}
B(t)=Y(t)-\int_0^t {\tr} \, (L\rho(s)+\rho(s) L^*) \, ds
\end{equation}
 the equation for the inverse trace rewrites as
  \begin{equation}
\label{eqtrinnov}
d\frac{1}{{\tr} \, \ga(t)}
=-\frac{1}{{\tr} \, \ga(t)} \, {\tr} \, (L\rho(t) +\rho(t) L^*) dB(t)
\end{equation}
and the equation for the normalized density operator \eqref{Lindstochnorm} rewrites in the
standard form \eqref{Lindstochnorm1} of the nonlinear filtering equation. Of course, these manipulations being straightforward in finite-dimensional case, may not be well justified in infinite-dimensions. However, the results of  Theorem \ref{LindStochLin1} make them all rigorous for the class of $C$-solutions. Notice that, for positive solutions, ${\tr} \, \ga(t)$ is a positive martingale that specifies the density between the measures $\P$ and $\Q$ that make $Y(t)$ and $B(t)$ Brownian motions, respectively:
\[
\E_{\Q}\xi=\E_{\P}(\xi \, {\tr}\, \ga(t))
\]
for $\FC_t$-adapted $\xi$. And therefore estimate \eqref{Lindstochquadex1}
for the solutions $\ga(t)$ of a linear equation with an initial condition of unit trace rewrites as the estimate
 \begin{equation}
\label{Lindstochquadex11}
\E_Q \|\rho(t)\|_C 
\le e^{\al t} [\|\ga_0\|_C +\al t (1+\be)],
\end{equation}
for the solution $\rho(t)=\ga(t)/{\tr}\, \ga(t)$ of the nonlinear equation.
Therefore, Theorem \ref{LindStochLin1} implies the existence of $C$-weak solutions to nonlinear equation \eqref{Lindstochnorm1}. 

Moreover, let a continuous pair of processes $(\rho(t), B(t))$ be a $C$-weak solution to \eqref{Lindstochnorm1} with a positive initial condition
$\rho_0$ and with ${\tr}\, \rho(t)=1$ for all $t$. Then, 
if the process ${\tr }\, \ga(t)$ is defined as the solution of equation \eqref{eqtrinnov} (with any positive initial condition), the process
$\ga(t)$ is defined as $\ga(t)= \rho(t) {\tr }\, \ga(t)$ and the process
$Y(t)$ is defined via \eqref{outputinnovation}, then these processes satisfy
the linear equation \eqref{Lindstoch}. Therefore, any $C$-weak solution of the nonlinear problem can be obtained by normalization from a $C$-solution to the linear one. Hence the uniqueness for the latter implies the uniqueness (in distribution) for the former.  

Summarizing we can conclude that Theorem \ref{LindStochLin2} implies the following well-posedness result for $C$-weak solutions of nonlinear equation \eqref{Lindstochnorm1}.

\begin{theorem}
\label{LindStochnonLin}
Assume Hypothesis MR and A hold, and let $\rho_0\in \TC_C^{S+}(\HC)$ have unit trace.
Then there exists a unique in law $C$-weak solution of equation \eqref{Lindstochnorm1} in $\HC^2$ with the initial data $\rho_0$.
\end{theorem}

\begin{remark}
Assuming only Hypothesis MR we get, by Theorem \ref{LindStochLin1}, that uniqueness 
holds under additional assumption of respecting finite-dimensional approximations.    
\end{remark} 

To analyze the strong solutions of \eqref{Lindstochnorm1}, we shall use
the key observation made in \cite{K} that one can rewrite stochastic master
 equation for mixed states in the equivalent form that coincides with the corresponding equation for pure states in an appropriately chosen Hilbert space.

The following result is proved in the same way as 
 the corresponding result from \cite{K} devoted to bounded operators $L$. We sketch the proof for completeness.
 
 \begin{theorem}
\label{LindStochnonLin1}
Assume Hypothesis MR and A hold, and let $\rho_0\in \TC_C^{S+}(\HC)$ have unit trace.
Then, for a given BM $B(t)$, there exists a unique $C$-strong solution of equation \eqref{Lindstochnorm1} in $\HC^2$ with the initial data $\rho_0$.
\end{theorem}

\begin{proof}
According to calculations performed above, $\rho(t)$ solves \eqref{Lindstochnorm1}
 if and only if $\ga(t)=T(t) \rho(t)$ solves the equation
  \begin{equation}
\label{Lindstochnormlinmix}
d\ga(t)=-i[H, \ga(t)] dt +\LC_L \ga(t) dt +(L\ga(t)+\ga(t) L^*)
[dB(t)+ \pi(t) \, dt],
\end{equation}
with
\[
\pi(t)=T^{-1}(t)\, {\tr} \, (L\ga(t)+\ga(t) L^*).
\]

Expanding $\rho_0=\ga_0$ in a series
\[
 \ga_0=\sum_{k=1}^{\infty} p_k e_k\otimes \bar e_k
 \]
 with a non-negative sequence $\{p_k\}$ summing up to one
 and an orthonormal basis $\{e_k\}$,
 we represent $\ga(t)$ as the convergence series of pure states
 \[
 \ga(t)=\sum_{k=1}^{\infty} p_k e_k(t)\otimes \bar e_k(t),
 \]
 with $e_k(t)$ solving the linear filtering equation for
pure states \eqref{eqqufiBlins}:
\begin{equation}
\label{eqinfdim}
de_k(t)=(-iH e_k(t)-\frac12 L^*L e_k(t))\,dt +Le_k(t) [dB(t)+\pi(t)\, dt],
\end{equation}
where
\[
\pi(t)=\frac{\sum_{k=1}^{\infty} p_k (e_k(t), (L+L^*) e_k(t))}{\sum_{k=1}^{\infty} p_k \|e_k(t)\|^2}.
\]

And now we observe that system \eqref{eqinfdim} can be equivalently written as a single SDE
with values in the  Hilbert space $l^2_{\HC}(\{p_k\})$ consisting of infinite sequences
$\e=(e_1, e_2, \cdots )$ of vectors from $\HC$ and equipped
with the norm
\[
\|\e\|^2=\sum_{k=1}^{\infty} p_k (e_k, e_k).
\]
The controlling operator $C$ and other operators involved (like $H$ and $L$)
 extend naturally (acting identically on each coordinate)
to operators in $l^2_{\HC}(\{p_k\})$.
In this notation system \eqref{eqinfdim} writes down as the SDE
\begin{equation}
\label{eqinfdim1}
d \e(t)=(-iH \e(t)-\frac12 L^*L \e(t))\, dt +L\e(t) 
\left[dB(t)+\frac{(\e, (L+L^*) \e)}{(\e, \e)} \, dt\right].
\end{equation}

This equation is the same as \eqref{eqqufiBlinsB} (though written in an enhanced Hilbert space). Hence the well-posedness for this equation follows from the well-posedness of the quantum filtering equation for pure states, that is Theorems \ref{thMorRebLin}
and \ref{strongpure}.
\end{proof}


\begin{thebibliography}{99}


\bibitem{Armen02Adaptive}
M. A. Armen, J. K. Au, J. K. Stockton, A. C. Doherty and H. Mabuchi. Adaptive homodyne measurement of optical phase. Phys. Rev. Let. 89 (2002), 133602.

\bibitem{BarbRock16}
V. Barbu, M R\"ockner and D. Zhang. Stochastic nonlinear Schrödinger equations.
Nonlinear Anal. 136 (2016), 168 - 194.

\bibitem{Barch24}
A. Barchielli. Markovian dynamics for a quantum/classical system
and quantum trajectories. J. Phys. A: Math. Theor. 57 (2024), 315301

\bibitem{BarchBel}
A. Barchielli and V.P. Belavkin. Measurements continuous in time and a posteriori states in quantum mechanics.
J. Phys A: Math. Gen. 24 (1991), 1495-1514.

\bibitem{BarchHol}
A. Barchielli and A.S. Holevo.
Constructing quantum measurement processes via classical
stochastic calculus. Stochastic Processes and their Applications 58 (1995) 293 - 317.

\bibitem{BarchBook}
A. Barchielli and M. Gregoratti. Quantum Trajectories and Measurements in Continuous Case. The Diffusive Case.
 Lecture Notes Physics, v. 782, Springer Verlag, Berlin, 2009.

\bibitem{BarndLoub}
O. E. Barndorff-Nielsen and E. R. Loubenets.
General framework for the behaviour of continuously
observed open quantum systems.
J. Phys. A: Math. Gen. 35 (2002) 565 - 588.

\bibitem{Bel87}
 V. P. Belavkin, Nondemolition measurement and control in quantum dynamical systems.
 In: Information Complexity and Control in Quantum Physics.
CISM Courses and Lectures 294, S. Diner and G. Lochak, eds., Springer-Verlag, Vienna, 1987, pp. 331–336.

\bibitem{Bel88}
V.P. Belavkin. Nondemolition stochastic calculus in Fock space and nonlinear
filtering and control in quantum systems. Proceedings XXIV Karpacz winter school
(R. Guelerak and W. Karwowski, eds.), Stochastic methods in mathematics and physics.
 World Scientific, Singapore, 1988, pp. 310 - 324.

\bibitem{Bel92} V.P. Belavkin. Quantum stochastic calculus and quantum nonlinear filtering.
 J. Multivar. Anal. 42 (1992), 171 - 201.

\bibitem{BelKol91}
V.P. Belavkin, V.N. Kolokoltsov.
 Quasyclassical asymptotics of quantum stochastic equations. 
 Teoret. i Matem. Fizika {\bf 89:2} (1991), 163-177.
Engl. transl. in Theor. Math. Physics.


 \bibitem{BelKol}
V.P. Belavkin, V.N. Kolokoltsov. Stochastic
evolution as interaction representation of a boundary value
problem for Dirac type equation. Infinite Dimensional Analysis,
Quantum Probability and Related Fields {\bf 5:1} (2002), 61-92.




\bibitem{BoutHanJamQuantFilt}
L. Bouten, R. Van Handel and M. James. An introduction to quantum filtering.
SIAM J. Control Optim. 46:6 (2007), 2199-2241.


 \bibitem{Bushev06Adaptive}
 P. Bushev et al. Feedback cooling of a singe trapped ion. Phys. Rev. Lett. 96 (2006), 043003.

\bibitem{Cheb89}
A.M. Chebotarev. Sufficient conditions for dissipative dynamic 
semigroups to be conservative.
Theor. Math. Phys. 80:2 (1989), 192-211.


\bibitem{ChebFagn}
A.M. Chebotarev and F. Fagnola. Sufficient conditions for conservativity of
minimal quantum dynamical semigroups. J. Funct. Anal. 153 (1998), 382 - 404.

\bibitem{ChebQuez}
A.M. Chebotarev, J. Garcia and R. Quezada.
 A priori estimates and existence
theorems for the Lindblad equation with unbounded time-dependent coefficients. In Recent
Trends in Infinite Dimensional Non-Commutative Analysis 1035 (1998), 44–65. Publ. Res.
Inst. Math. Sci., Kokyuroku, Japan.

\bibitem{Dios}
L. Diosi.
Hybrid completely positive Markovian quantum-classical dynamics.
 Phys. Rev. A 107 (2023), 062206.

\bibitem{Fagnola}
F. Fagnola and C. M. Mora.
Stochastic Schr\"odinger equation and applications to Ehrenfest-type theorems.
  ALEA Lat. Am. J. Probab. Math. Stat. 10:1 (2013), 191 - 223.


\bibitem{Holevo91} A.S. Holevo. Statistical Inference for quantum processes.
 In: Quanum Aspects of Optical communications.
Springer LNP 378 (1991), 127-137, Berlin, Springer.

 \bibitem{Holevo96} A.S. Holevo.
 On dissipative stochastic equations
in a Hilbert space. Probab. Theory Relat. Fields 104 (1996), 483 - 500.

\bibitem{Kolok95}
V.N. Kolokoltsov.
 Scattering theory for the Belavkin equation describing
a quantum particle with continuously observed coordinate.
 Journ. Math.Phys. 36:6 (1995), 2741-2760.


\bibitem{KolbookSemi}
V. N. Kolokoltsov. Semiclassical analysis and diffusion processes.
Springer Lecture Noted in Math. v. 1742, Springer 2000.



\bibitem{KolQuantMFGCount}
V. N. Kolokoltsov.
Quantum Mean-Field Games with the Observations of
Counting Type.  Games (2021), 12, 7.

\bibitem{KolQuantLLN}
V. N. Kolokoltsov.
The law of large numbers for quantum stochastic filtering 
and control of  many particle systems.
 Theoretical and Mathematical Physics 208:1 (2021), 97-121. 
 English translation 208(1), 937-957.

\bibitem{KolQuantMFG}
V. N. Kolokoltsov. Quantum Mean Field Games.
Annals Applied Probability 32:3 (2022), 2254 - 2288.

\bibitem{KolQuantFrac}
V. N. Kolokoltsov.
Continuous time random walks modeling of quantum measurement and
fractional equations of quantum stochastic filtering and control.
Fractional Calculus and Applied Analysis 25 (2022), 128 - 165.

\bibitem{K} V.N. Kolokoltsov. On quantum stochastic master equation.
Electron. J. Probab. 30 (2025), PNO 57, 1-21.


\bibitem{Loubenets}
E. R. Loubenets.
Quantum Stochastic Approach to the Description of
Quantum Measurements.
J. Phys. A: Math. Gen. 34 (2001), 7639

 \bibitem{Pellegrini} C. Pellegrini.
Poisson and Diffusion Approximation of Stochastic Schr\"odinger Equations with Control.
 Ann. Henri Poincar\'e 10:5 (2009), 995–1025.

 \bibitem{Pellegrini10} C. Pellegrini.
Markov chains approximation of jump–diffusion stochastic
master equations.
 Ann. Henri Poincar\'e 46: 4 (2010), 924–948.

\bibitem{Mora13}
C. M. Mora. Regularity of solutions to quantum master equations: a stochastic approach.
The Annals of Probability
41:3B (2013), 1978 - 2012.

\bibitem{MoraRebolin}
C. M. Mora and R. Rebolledo.
Regularity of solutions to linear stochastic
Schr\"odinger equations.
Infinite Dimensional Analysis Quantum Probability and Related Topics · June 2007
DOI: 10.1142/S0219025707002725


\bibitem{MoraRebo}
C. M. Mora and R. Rebolledo.
Basic Properties of Nonlinear Stochastic Schrödinger Equations Driven by Brownian Motions.
 The Annals of Applied Probability 18:2 (2008), 591 - 619.

\bibitem{Partha}
K. R. Parthasarathy and A. R. Usha Devi. 
From quantum stochastic differential equations to Gisin-Percival state diffusion.
J. Math. Phys. 58, 082204 (2017).

\bibitem{VanNeerven}
J. van Neerven, M. Veraar and L. Weis (2015).
Stochastic Integration in Banach Spaces – a Survey.
In: Dalang, R., Dozzi, M., Flandoli, F., Russo, F. (eds) Stochastic Analysis:
A Series of Lectures. Progress in Probability, vol 68. Birkhäuser, Basel.

 \bibitem{WiMilburnBook}
H. M.  Wiseman and G. J.  Milburn.
 Quantum measurement and control. Cambridge Univesity Press, 2010.

\end{thebibliography}
\end{document}